\documentclass[AMA,STIX2COL,Linenumbersoff]{MRM}
\usepackage{amsmath}        
\usepackage{amssymb}        
\usepackage{bm}             

\articletype{Research Article}

\received{xx xx 2024}
\revised{x xx 2025}
\accepted{x xx 2025}
\renewcommand{\footnote}[1]{}  

\makeatletter
\def\orcid#1{} 
\def\orcid#1{\href{#1}{\includegraphics{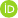}}} 
\makeatother

\begin{document}

\title{Whole Heart Perfusion with High-Multiband Simultaneous  Multislice  Imaging  via  Linear  Phase  Modulated  Extended Field of  View (SMILE)}

\author[1,3]{Shen Zhao}{\orcid{0000-0003-1736-6378}}
\author[1]{Junyu Wang}{\orcid{0000-0001-8314-4525}}
\author[1,3]{Xitong Wang}{\orcid{0009-0002-8457-9189}}
\author[1,3]{Sizhuo Liu}{\orcid{0000-0002-3248-0737}}
\author[1,3]{Quan Chen}{\orcid{0000-0002-5392-8000}}
\author[1]{Kevin Kai Li}{}
\author[3]{Yoo Jin Lee}{\orcid{0009-0008-2348-2520}}
\author[1,2,3,4]{Michael Salerno}{\orcid{0000-0001-7051-1031}}

\authormark{Shen Zhao \textsc{et al}}

\address[1]{\orgdiv{Department of Cardiovascular Medicine}, \orgname{Stanford University}, \orgaddress{\state{CA}, \country{USA}}}
\address[2]{\orgdiv{Department of Radiology}, \orgname{Stanford University}, \orgaddress{\state{CA}, \country{USA}}}
\address[3]{\orgdiv{Department of Radiology \& Biomedical Imaging}, \orgname{University of California, San Francisco}, \orgaddress{\state{CA}, \country{USA}}}
\address[4]{\orgdiv{Department of Medicine, Division of Cardiology}, \orgname{University of California, San Francisco}, \orgaddress{\state{CA}, \country{USA}}}

\corres{Michael Salerno, \email{msalerno@stanford.edu, 	michael.salerno@ucsf.edu}}

\presentaddress{35 Medical Center Way, M314A
San Francisco CA 94143}

\finfo{This work was partially supported by \fundingAgency{National Institutes of Health} grant \fundingNumber{R01HL131919, R01HL155962-01}}

\abstract{
\section{Purpose} To develop a simultaneous multislice  (SMS) first-pass perfusion technique that can achieve whole heart coverage with high multi-band factors, while avoiding the issue of slice leakage.
\section{Methods} The proposed Simultaneous Multislice Imaging via Linear phase modulated Extended field of view (SMILE) treats the SMS acquisition and reconstruction within an extended field of view framework, allowing arbitrarily under-sampling of phase encoding lines of the extended k-space matrix and enabling the direct application of 2D parallel imaging reconstruction techniques. We presented a theoretical framework that offers insights into the performance of SMILE. We performed retrospective comparison on 28 subjects and prospective perfusion experiments on 43 patients undergoing routine clinical CMR studies with SMILE at multiband (MB) factors of 3-5, with a net acceleration rate ($R$) of 8 and 10 respectively, and compared SMILE to conventional SMS techniques using standard FOV 2D CAIPI acquisition and standard 2D slice separation techniques including split-slice GRAPPA and ROCK-SPIRiT.

\section{Results} Retrospective studies demonstrated 5.2 to 8.0 dB improvement in signal to error ratio (SER) of SMILE over CAIPI perfusion. Prospective studies showed good image quality with grades of 4.1 $\pm$ 0.7 for MB = 3, $R$ = 8 and 3.5 $\pm$ 1.0 for MB = 5, $R$ = 10. (5-point Likert Scale)

\section{Conclusion} The theoretical derivation and experimental results validate the SMILE's improved performance at high acceleration and MB as compared to the existing 2D CAIPI SMS acquisition and reconstruction techniques for first-pass myocardial perfusion imaging.
}
\wordcount{4760}


\maketitle

\section{Introduction}
Cardiac magnetic resonance imaging (CMR) vasodilator stress myocardial perfusion imaging has demonstrated important diagnostic and prognostic utility for the evaluation of ischemic heart disease\cite{canet1999magnetic}.  First-pass perfusion is one of the most demanding CMR applications as multiple saturation-recovery images must be acquired in each heartbeat at stress heart rates, which can exceed 110 beats per minute.  This typically results in partial coverage of the ventricle to 3-4 slices with $2\times 2$ mm$^2$ spatial resolution with most clinically available pulse sequences. 

To address these challenges, researchers have explored highly accelerated spatiotemporal undersampling paired with advanced reconstruction techniques like compressed sensing \cite{otazo2010combination}. Additionally, to further enhance efficiency, non-Cartesian acquisition strategies such as spiral \cite{yang2019whole} and radial \cite{kholmovski2007perfusion} trajectories have been investigated. These approaches are often combined with the aforementioned acceleration techniques to optimize image quality and acquisition speed \cite{sharif2014towards, yang2016first, shin2013three}. To reduce the overhead of separate saturation preparation blocks for each slice, multiple investigators have also pursued simultaneous multislice (SMS) techniques \cite{yang2019whole, wang2016radial, mcelroy2022simultaneous, sun2022slice, demirel2023signal} including non-Cartesian approaches to improve acquisition efficiency to improve slice coverage to 6-9 slices per heartbeat, enabling extended ventricular coverage.  With most techniques to date, SMS perfusion has been limited to multiband (MB) factors of 2-3.

Various SMS acquisition and reconstruction techniques have been proposed. Acquisition strategies include Hadamard \cite{souza1988sima}, POMP \cite{glover1991phase}, SMS (w/o phase modulation) \cite{larkman2001use}, CAIPIRINHA \cite{breuer2005controlled}, and 3D acquisition \cite{zahneisen2014three, zhu2016hybrid}. SMS-specific reconstruction algorithms include SENSE-GRAPPA \cite{blaimer2006accelerated}, Slice-GRAPPA (SG) \cite{setsompop2012blipped}, Split Slice-GRAPPA (SPSG) \cite{cauley2014interslice}, and ROCK-SPIRiT (RS) \cite{demirel2021improved}.

A number of these approaches have been applied to first-pass perfusion imaging. A significant limitation of SMS imaging has been slice leakage, which limits image quality. Slice leakage results from interference from multiple simultaneously excited slices.  Most techniques have sought to use kernel-based approaches to cancel signal interference from adjacent slices.

Our work \cite{zhao2024whole} analyzed and generalized the POMP extended FOV concept and proposed a concise accelerated SMS acquisition and reconstruction framework: Simultaneous Multislice Imaging via Linear phase modulated Extended field of view (SMILE). SMILE transforms the SMS problem into a 2D parallel imaging task, i.e., SMILE transforms the superimposed slice separation problem into a 2D de-aliasing problem which enables the direct application of parallel imaging reconstruction algorithms and can theoretically avoid the slice leakage. For SMILE, the spatial extent of the extended FOV does not need to be a multiple of the excited number of slices, i.e., the MB factor, and the net acceleration rate $R$ is solely determined by the net in-plane acceleration relative to the extended FOV, which is different from traditional SMS techniques and POMP.  Furthermore, by choosing an appropriate spatial-temporal acquisition, SMILE is auto-calibrating without requiring separate acquisition of the individual slices.  

In this manuscript, we evaluated the performance of the SMILE framework for first-pass perfusion images, with MB = 3 and 5, and a 3$\times$ or 5$\times$ extended FOV respectively, comparing it to traditional slice-separation approaches such as SPSG and RS with MB factors of 3 or 5, which rely on CAIPI phase modulation within a single-slice FOV.  For a fair comparison both CAIPI and SMILE techniques had the same net acceleration, and thus acquired the same number of phase encoding (PE) lines. This evaluation included both retrospective downsampling and prospective undersampling in in-vivo perfusion studies, demonstrating slice-leakage-free reconstruction and improved performance. The prospective studies achieved whole-heart coverage with net acceleration rates  $R$ of up to 10, providing whole heart coverage with $1.5 \times 1.5$ mm$^2$ in-plane resolution using 2 to 3 saturation recovery blocks.

\section{Theory}\label{sec2}
\label{subsection: SMILE}
Although different perspectives can be employed in implementing SMS acquisition or reconstruction including the stack of 2D, 3D \cite{zahneisen2014three, zhu2016hybrid}, extended 2D FOV along readout (RO) \cite{demirel2021improved}, and extended 2D FOV along PE \cite{glover1991phase}, the information contained in the 2D slice images is identical, as they are imaging the same real object. Despite this theoretical equivalence, acquisition strategies and reconstruction performance differences may still arise. For instance, the discontinuity in the inter-slice coil sensitivity maps, in contrast to conventional 3D imaging, impedes the construction of an appropriate 3D k-space reconstruction kernel.  Differences in sampling strategy for a give reconstruction approach can also impact the conditioning number of the reconstruction, impacting reconstruction quality.  The extended 2D FOV along RO perspective and 2D slice separation approaches using conventional phase modulation cannot sample all k-space points for the extended FOV and thus have a limited sampling degree of freedom.

\subsection{SMILE Acquisition}
Following our subsequent analysis, we present SMILE, a solution that can overcome the intrinsic limitations of current methods, facilitating the direct incorporation of standard 2D parallel imaging reconstruction techniques. The SMILE acquisition has two main steps:
\begin{itemize}
    \item Adopt an extended FOV along PE to encompass all slice image regions of interest without overlapping. Note, the extended FOV does not necessarily have to be a multiple of the MB factor.
    \item Phase modulate PE lines of each slice linearly to distribute slices' ROI within the extended FOV. Note that different slices can be shifted arbitrarily within the extended FOV.
\end{itemize}
Slice phase modulation can be achieved using RF excitation phases or gradient blips, with the latter being restricted to inter-slice phase differences linearly proportional to inter-slice distances. RF phase modulation, given its higher degree of freedom, allows for arbitrary inter-slice phase differences. Notably, SMILE phase modulation need not be restricted to CAIPI; any RF modulation that shifts slice images without ROI overlap is possible. SMILE permits full sampling and arbitrary PE undersampling over the extended FOV, enhancing incoherent sampling for compressed sensing reconstruction versus superimposed single-slice FOV k-space acquisition \cite{larkman2001use, breuer2005controlled, demirel2021improved}. SMILE does not necessarily require separate calibration data. Autocalibration signal region or extra-dimensional averaging \cite{feng2014golden} allows valid extraction of linear prediction kernels or coil sensitivity maps. However, systematic studies investigating sampling and reconstruction from this perspective are lacking.

\subsection{SMILE Reconstruction}
Assuming the $c$-th continuous coil sensitivity map, $\bm{S}_c$, has a bandwidth size of $[B_x, B_y]$. As illustrated in Fig.~\ref{fig: CSM Support}, the discrete k-space bandwidth of the coil sensitivity map for a single-slice FOV, with a sampling grid of $\Delta k_x, \Delta k_y$, amounts to $[C_x, C_y] = \big[\lfloor \tfrac{B_x}{\Delta k_x} \rfloor, \lfloor \tfrac{B_y}{\Delta k_y} \rfloor\big]$. For an $n \times$ extended FOV along PE, this value changes to $[C_x, D_y] = \big[\lfloor \tfrac{B_x}{\Delta k_x} \rfloor, \lfloor \tfrac{n B_y}{\Delta k_y} \rfloor\big]$.

\begin{figure}[htbp]
    \centering
    \includegraphics[width = \columnwidth]{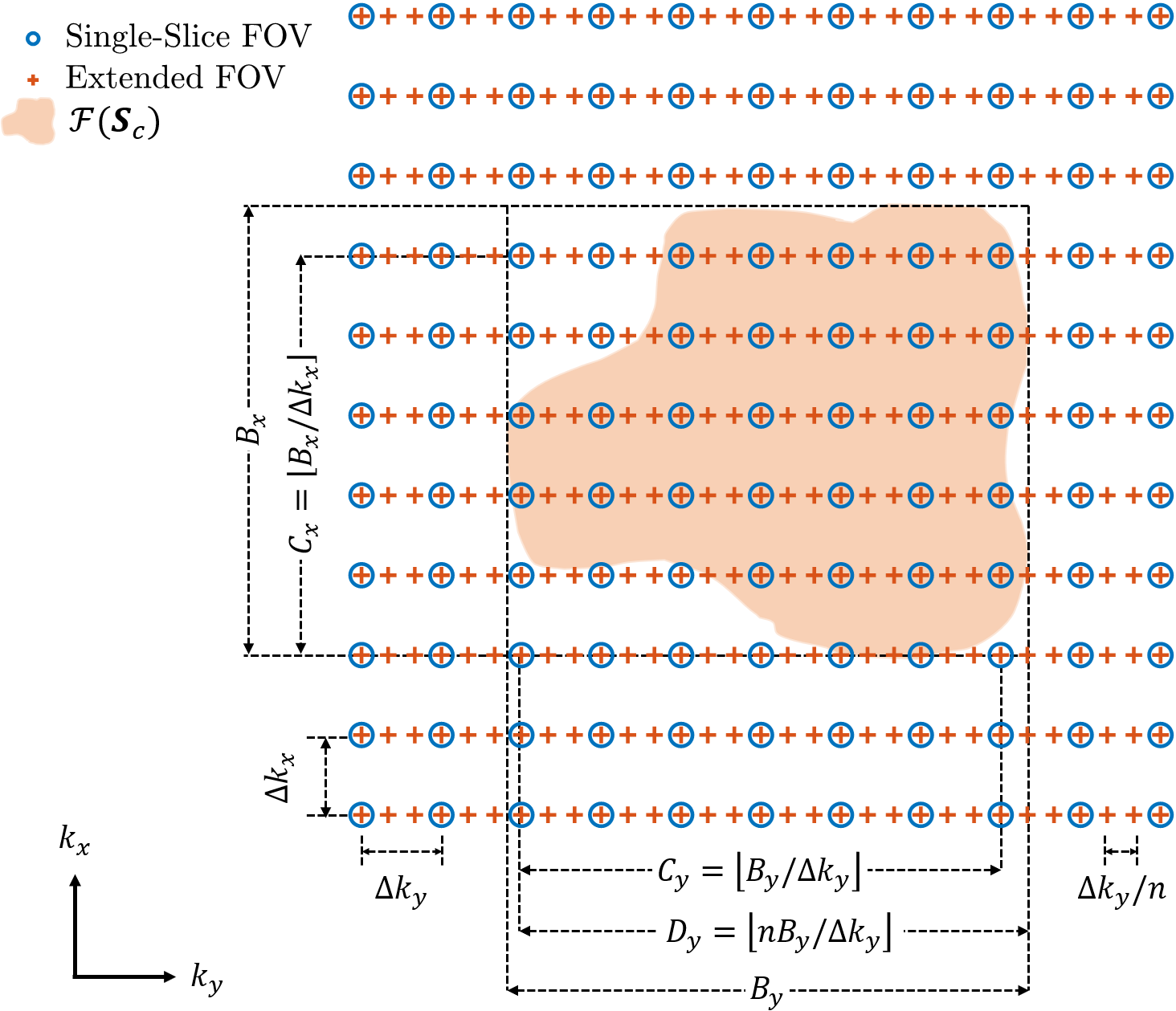}
    \caption{Discrete support of coil sensitivity map $\bm{S}_c$ in different k-space sampling grid sizes $\Delta k_x, \Delta k_y$ or $\Delta k_y/n$.}
    \label{fig: CSM Support}
\end{figure}

This illustration has an important implication for choosing an appropriate kernel size for SMILE.  For instance, let $n \in \mathbb{Z}_{>0}$, $C_y = \lfloor \frac{D_y}{n} \rfloor >> 1$. If, for a single-slice field of view (FOV), the smallest reconstruction kernel size based on independent band-limited coil sensitivity maps is $[E_x^\star, E_y^\star] = [6,6]$, then for $n\times$ FOV along PE, the kernel size scales approximately as $[E_x^\star, E_y^\star] \approx [6,6n]$. This is rather atypical for 2D parallel imaging \cite{shin2014calibrationless, gungor2014subspace, griswold2002generalized, lustig2010spirit, zhang2011parallel,haldar2013low, zhao2021high, zhao2022high}. Using an insufficient kernel size can impact reconstruction performance. A detailed theoretical justification is provided in the Supporting Information.

In SMILE discrete coil sensitivity map extraction, referencing Fig. \ref{fig: CSM Support}, both the discrete calibration and auto-calibration signal (ACS) regions need to expand along the PE dimension (due to smaller $\Delta k_y$ to support the extended FOV) to adequately capture coil sensitivity maps. This up-scaled kernel size should be integrated into methods like ESPIRiT \cite{uecker2014espirit}, for which when slice overlaps occur, multiple sensitivity map sets might be needed.

Our analysis reveals that SMILE allows for any compatible 2D parallel imaging reconstruction methods (SENSE or k-space kernel-based), making SMS-specific methods like SG, SPSG, and RS unnecessary for SMILE. Other MR physics priors, such as smooth-phase constraints, can also be incorporated into SMILE reconstruction if inter-slice and intra-slice phases are both smooth \cite{lobosadvanced}.

\subsubsection{Content Dependence and Slice Leakage Analysis}
Innately, SENSE-based reconstruction can bypass bias and effectively remove dependence on reconstruction content, unlike k-space-based reconstruction. K-space linear predictability for 2D parallel imaging mainly stems from coil sensitivity map smoothness and limited image support \cite{haldar2020linear}. Despite the lack of a theoretical guarantee on GRAPPA's linear prediction kernels depending solely on coil sensitivity maps, especially with limited image support, past studies and practical GRAPPA implementation affirm its robust performance and approximate content-independent kernel across various scenarios \cite{setsompop2012blipped}. This suggests the selected kernel size $[E_x, E_y]$ adequately leverages bandlimited $\bm{S}_c$ prior yet may fall short for the limited image support prior. When all slice images in the single-slice FOV possess this characteristic, the potential exists for finding content-independent kernels for the extended 2D FOV along PE under suitable conditions.

As per analysis, to capture coil sensitivity map smoothness, the smallest kernel for each slice should increase along PE by approx $\frac{D_y-1}{C_y-1}\approx n$. Concurrently, limited image support for each slice within the single-slice FOV, corresponds to a kernel magnified by $n$ along PE in the extended FOV, according to Fourier transform properties of zero padding. Consequently, if the smallest linear prediction kernel for each slice in a single-slice FOV is content-independent, then it is likely that the smallest kernel for an extended FOV—whether with narrow, or no inter-slice gaps, or even partial overlaps—will also remain content-independent. This is because the extended FOV will not significantly increase the size of the connected empty background region along the PE direction. (see Fig. 4 in reference \cite{haldar2020linear} for an illustration of the empty region and linear prediction). As a result, achieving content independence through proper SMILE acquisition becomes simpler.


As emphasized in the review \cite{moeller2021diffusion}, concepts like slice leakage or ``slice blocking'' are not applied to conventional parallel imaging but arise from SG or SPSG implementation due to their bias and kernel content-dependence vulnerability. Notably, SMILE, like parallel imaging, does not factor in slice leakage or blockage. Thus, leakage assessment tools like the linear system leakage approach (LSLA) \cite{cauley2014interslice} are not fit for SMILE evaluation. Assuming the SMILE reconstruction kernel size is sufficient, the resultant reconstruction should be unbiased, affected only by noise perturbations. When reconstructions struggle at high acceleration rate $R$, one can quantify residual aliasing using image quality metrics such as mean squared error.

\subsubsection{Sampling Pattern and g-Factor Analysis}
Unlike single-slice parallel imaging, where undersampling causes partial overlying of image content's alias onto itself, undersampling in SMILE extended FOV might result in specific slices directly overlapping others. This is more likely when all slice supports are identical and uniformly dispersed in the extended FOV with uniform sampling at $R=$ integer multiple of the MB factor. If $R$ is co-prime to the MB factor, the situation resembles a single-slice case, with a slice alias partially overlapping another. While optimal sampling for SMILE extended FOV remains under investigation, we can identify some worst-case SMILE sampling scenarios such as when $R=$ integer multiple of the MB factor, and slices' coil sensitivity maps are very similar.

For illustration, we used a genetic algorithm (GA) to optimize sampling for minimal average g-factor, utilizing one frame of a fully sampled OCMR short-axis cardiac cine dataset with MB = 3, uniformly distributed across the $n=3\times$ extended FOV at various accelerations \cite{chen2020ocmr}. The initial population of 50 included uniform sampling, uniform density Poisson sampling, CAVA variable density sampling \cite{rich2020cartesian}, and random sampling patterns. CAVA variable density sampling produces a sequence of PE indices based on the golden ratio increment and has the benefit of the capability of retrospective adjustment of temporal resolution, maintaining incoherence, and ensuring a fully sampled time-averaged k-space to facilitate sensitivity map estimation.  The g-factor w.r.t. SENSE reconstruction was calculated using the pseudo multiple replica method \cite{robson2008comprehensive} with 64 Monte Carlo simulations. Results are shown in Fig.~\ref{fig: sampling PSF g-factor}. For $R=3$ and $R=6$ (integer multiples of MB), the GA mask significantly differs from uniform sampling, trending towards Poisson sampling at $R=3$ and the CAVA pattern at $R=6$. In these cases $R$ = integer multiple of MB, which means other slice content will directly overlap onto one slice, making the reconstruction likely ill-conditioned. For $R=4,5,7,8$ (relatively co-prime to MB = 3), the GA sampling pattern resembles uniform sampling for $R=4,5$, but shifts towards a variable density, similar to CAVA, at higher accelerations ($R=7,8$). As illustrated by the point spread function magnitude, $R=4,5$ exhibit $R-1$ side lobes, while $R=3,6,7,8$ present a single peak. 

Note that this optimized sampling pattern is specific to SENSE reconstruction. For other reconstruction methods, such as kernel-based techniques, spatial-temporal joint reconstructions using compressed sensing, global or local low-rankness, and machine learning, the optimal sampling pattern is still under investigation.

\begin{figure*}[htbp]
    \centering
    \includegraphics[width = \textwidth]{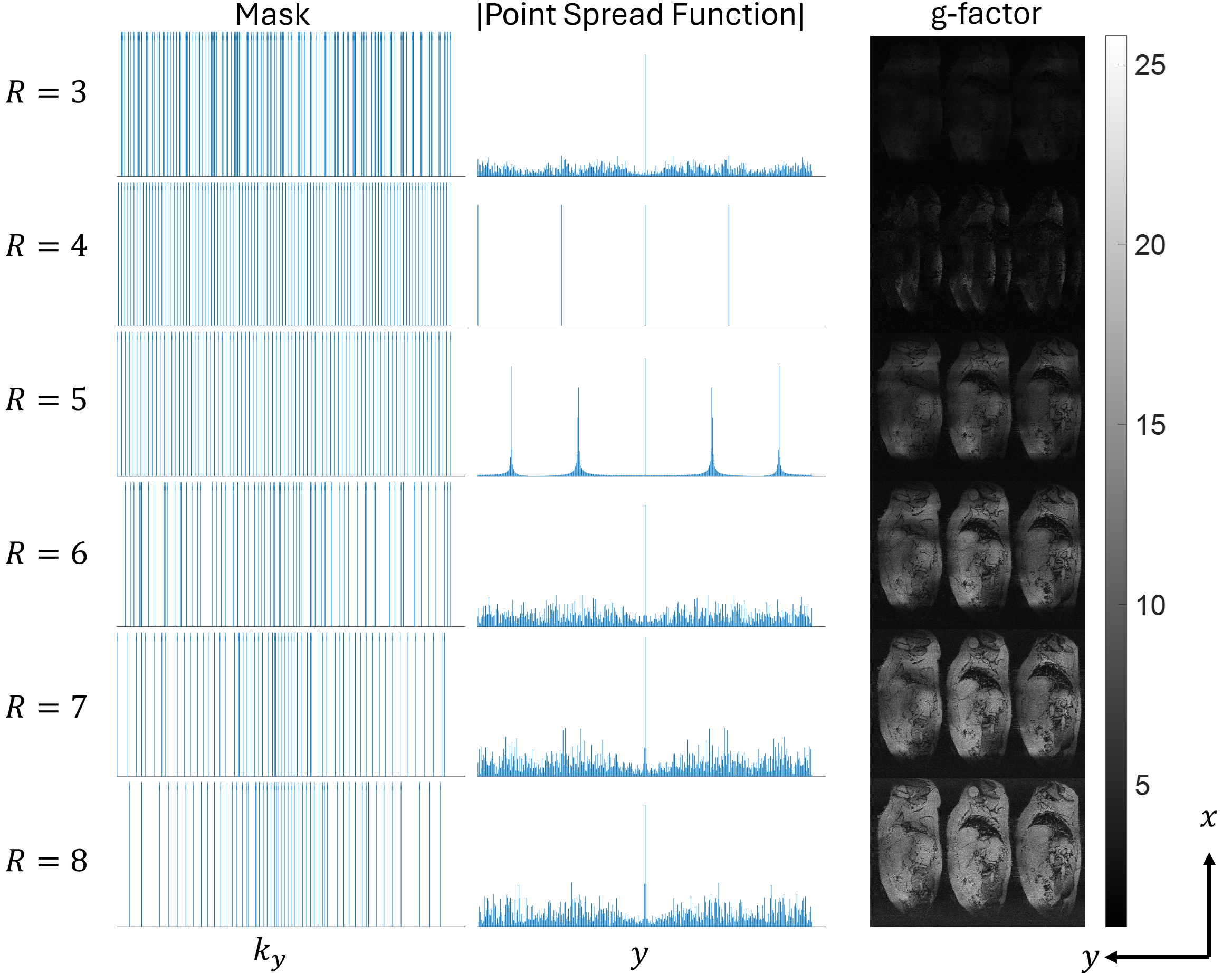}        
    \caption{Left to right: GA optimized sampling pattern, the absolute value of the point spread function, and the g-factor calculated using a pseudo multiple replica method. MB = 3 uniformly distributed slices in the $n=$MB$=3$ extended FOV pertain to an OCMR short-axis cardiac cine dataset. The GA involves 50 generations of population size 50 and includes uniform, Poisson, and CAVA variable density sampling in the first generations. Both the GA and g-factor calculations employ SENSE reconstruction. We conduct 64 Monte Carlo simulation trials for g-factor calculation. Empty background areas are cropped for space efficiency.}
    
    \label{fig: sampling PSF g-factor}
\end{figure*}

\section{Methods}
\paragraph{Retrospective Downsampled Perfusion Study}
Resting perfusion scans from 28 subjects undergoing clinically ordered CMR were utilized for this study under an IRB-approved protocol. Images were acquired with free-breathing during intravenous injection of 0.075 mmol/kg of Gadovist at 3 cc/second over 60 heartbeats in a 3T scanner (Siemens MAGNETOM Skyra). Each study included 3 to 5 short-axis perfusion slice locations.  Images were acquired with the following sequence parameters: GRE with separate calibration data, resolution 2.1$\times$2.4 mm$^2$, $R=2$ uniform sampling, SR preparation 5 pulse train, TI 100 ms, TR 2.63 ms, TE 1.17 ms.  Images were reconstructed with per frame HICU (HICU 2D) \cite{zhao2021high} to serve as the reference. Before downsampling, the datasets were compressed to 9 coils using SVD coil compression. We selected the same net acceleration rate $R$ for both CAIPI and SMILE acquisitions, meaning that the total number of PE lines acquired was identical for both techniques for all comparisons. For SMILE, the net acceleration rate $R$ corresponds to the in-plane acceleration rate relative to the fully sampled extended FOV for a given resolution, regardless of the MB factor; whereas for CAIPI, $R$ is the product of the in-plane acceleration rate and the MB factor. Note that the extended factor of FOV and MB factor do not need to be the same for SMILE.

For the MB = 3 experiments, we used a net $R=6$ for both CAIPI and SMILE acquisition and reconstruction on 16 datasets.  The SMILE and CAPI acquisitions and reconstruction had the same total number of PE lines.   For CAIPI, we employed an in-plane $R=2$ uniform sampling pattern with traditional phase modulation, while for SMILE, we used a $3\times$ extended FOV with a CAVA sampling pattern and uniform distributed slices in the extended FOV. This approach ensured that both CAIPI and SMILE acquisitions had the same total number of PE lines. We adopted CAVA sampling not only for its demonstrated effectiveness in 2D per-frame SMILE reconstruction but also for its additional benefit of enhancing incoherence in the transform temporal dimension.

Similarly, for the MB = 5 experiment, we used a net $R=5$ for both techniques on 12 datasets, with a $5\times$ extended FOV and uniform distributed slices in the extended FOV for SMILE and no in-plane downsampling for CAIPI. This choice was driven by the limited performance of conventional methods with a $2\times$ in-plane acceleration for MB = 5 CAIPI acquisition, which would result in a net $R=10$.

We include state-of-the-art reconstruction methods for different acquisitions to compare and estimate the approximate upper bound of performance for both CAIPI and SMILE acquisitions. The acquisition + reconstruction combinations includes CAIPI + \{SG (+ G), SPSG (+ G), RS\}, SMILE + HICU 2D. The (+ G) is the optional in-plane GRAPPA reconstruction after separating the slices, when there is in-plane undersampling. In the Supporting Information, we provide a spatial-temporal reconstruction comparison using different additional regularization methods. Specifically, we include SENSE with stationary wavelet transform (SWT) as the same spatial-temporal reconstruction approach for both SMILE and CAIPI acquisitions. This comparison aims to determine whether the  observed performance differences arise from the acquisition method (CAIPI vs SMILE), rather than from differences in the chosen reconstruction techniques.




All reconstructions were fine-tuned based on two separate datasets to maximize the SER. The kernel size in $k_x, k_y$ that maximized the reconstruction SER was chosen for MB = 3:
\begin{align*}    
    \text{SG + G}:& [13, 13], [17, 17] &    \text{SPSG + G}:&[17, 17], [17, 17]\\
    \text{RS}:& [33, 11]& \text{HICU 2D}:& [6, 18],    
\end{align*}
for $\text{MB} = 5$:
\begin{align*}
    \text{SG}:& [11, 11]& \text{SPSG}:&[19, 19]\\
    \text{RS}:& [27, 5]&\text{HICU 2D}:& [6, 30].
\end{align*}
Note that these experimental optimized kernel shapes and sizes are consistent with our analysis of the reconstruction kernel corresponding to the FOV. The reconstruction SER and SSIM were analyzed using repeated measure ANOVA.

\paragraph{Prospective Undersampled Perfusion Study}
13 resting perfusion scans with MB = 3, a 3$\times$ extended FOV (net $R=8$) and 30 resting perfusion scans with MB = 5, a 5$\times$ extended FOV (net $R=10$) were prospectively acquired using  SMILE to achieve whole left ventricle coverage (6 or 9 slices for MB = 3, 5 or 10 slices for MB =  5). Imaging was performed with breath-holding during an intravenous injection of 0.075 mmol/kg of Gadovist at 3 cc/second over 60 heartbeats in a 3T scanner (Siemens MAGNETOM Skyra). SMILE GRE sequence parameters included: TE: 1.29 ms, TR: 2.9 ms, Time of delay (TD) from SAT pulse to first PE line acquisition for Tissue function: 60 ms, resolution 1.5 $\times$ 1.5 $\times$ 10 mm$^3$, FOV 288 $\times$ 864 mm $^2$ for MB = 3 or 288 $\times$ 1440 mm$^2$ for MB = 5. We adopt CAVA sampling, net acceleration rate $R=8$ for MB = 3 (72 measured PE lines), and $R=10$ for MB = 5 (96 measured PE lines). Spatial-temporal HICU + stationary wavelet (SWT) was used as reconstruction. Images were graded by three cardiologists/radiologists using a 5-point Likert scale. The acquisition pipeline is shown in Fig.~\ref{fig: Acquisition pipeline}.
\begin{figure}[t]
    \centering
    \includegraphics[width= \columnwidth]{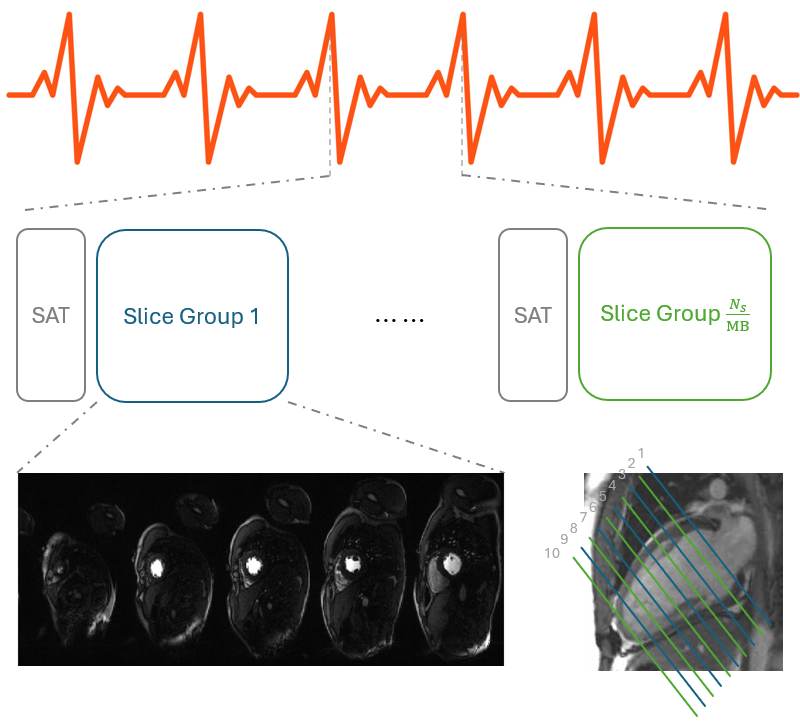}
    \caption{Acquisition pipeline of SMILE perfusion. Within each heartbeat, saturation recovery pulse train is sequentially applied for the interleaved slice group 1,$\cdots$, $N_s/\text{MB}$. This figure uses MB = 5, $N_s=10$ slices as illustration.}
    \label{fig: Acquisition pipeline}
\end{figure}

\section{Results}
Fig.~\ref{fig:Retrospective MB=3} and Fig.~\ref{fig:Retrospective MB=5} show the representative reconstruction and absolute value of error map, we can see that the CAPI + existing state-of-art methods show inferior image quality as compared to SMILE + reconstruction methods, where the slice leakage from other slices is more easily observed in the absolute error map in the CAPI results. We have also included 2D + t reconstruction and comparison movies in the Supporting Information.

The summarized radar plots for SER (dB) and SSIM are shown in Fig.~\ref{fig:SER SSIM}, where there is a 4.4 dB SER and 0.20 SSIM between all CAIPI + 2D reconstruction and SMILE + HICU 2D for MB = 5, $R$ = 5, and a 7.6 dB SER gap and 0.18 SSIM gap for MB = 3, $R$ = 6. Both show statistically significant differences ($p< 0.05$).  To assess the effects of the sampling strategy (CAIPI vs SMILE) we performed reconstruction using SENSE + SWT as an example spatial-temporal reconstruction that can be applied to both techniques.  In supplemental Fig.~S1, we show that a significant part of the advantage is related to the sampling strategy in SMILE which improved SER and SSIM by up to 8.0 dB and 0.12. The additional improvement in performance results from HICU as the reconstruction technique which provides high image quality even for higher acceleration rates.

Fig.~\ref{fig:Prospective MB=3} and Fig.~\ref{fig:Prospective MB=5} show the representative frame of the prospective experiment with good image quality without slice leakage. The grades of MB = 3, $R$ = 8 and MB = 5, $R$ = 10 are 4.1 $\pm$ 0.7 and 3.5 $\pm$  1.0.

Fig.~\ref{fig:Temporal fidelity} shows good agreement of signal intensity of spatial-temporal HICU + SWT reconstruction compared to per frame HICU 2D reconstruction.  This demonstrates that this spatiotemporal acceleration approach does not markedly impact temporal fidelity of the data.

Fig.~\ref{fig:SMILE resting perfusion defect and LGE} shows the detected SMILE resting perfusion defect, along with the corresponding late gadolinium enhancement (LGE) phase-sensitive inversion recovery (PSIR) images. The patients have diffuse LGE from amyloid. The SMILE resting perfusion defects show diffuse subendocardial hypoperfusion, which corresponds with the most severe regions of LGE. 

\begin{figure*}    
    \centering
    \includegraphics[width=\textwidth]{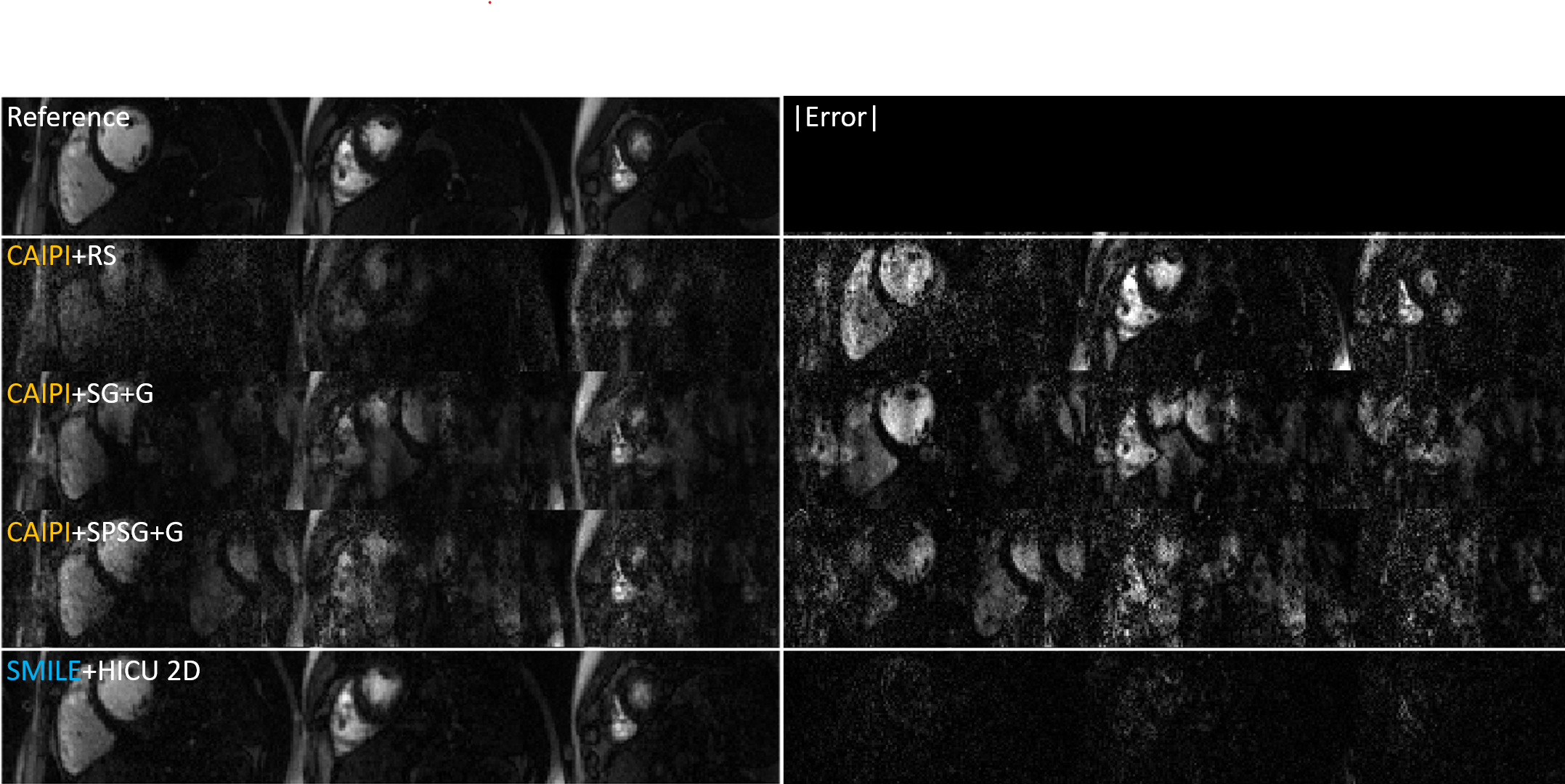}    
    \caption{Representative retrospective result for MB factor  = 3, $R$ = 6. The absolute error is scaled by 2, and windowed to match the reconstruction intensity range. }
    \label{fig:Retrospective MB=3}
\end{figure*}

\begin{figure*}
    \centering
    \includegraphics[width=\textwidth]{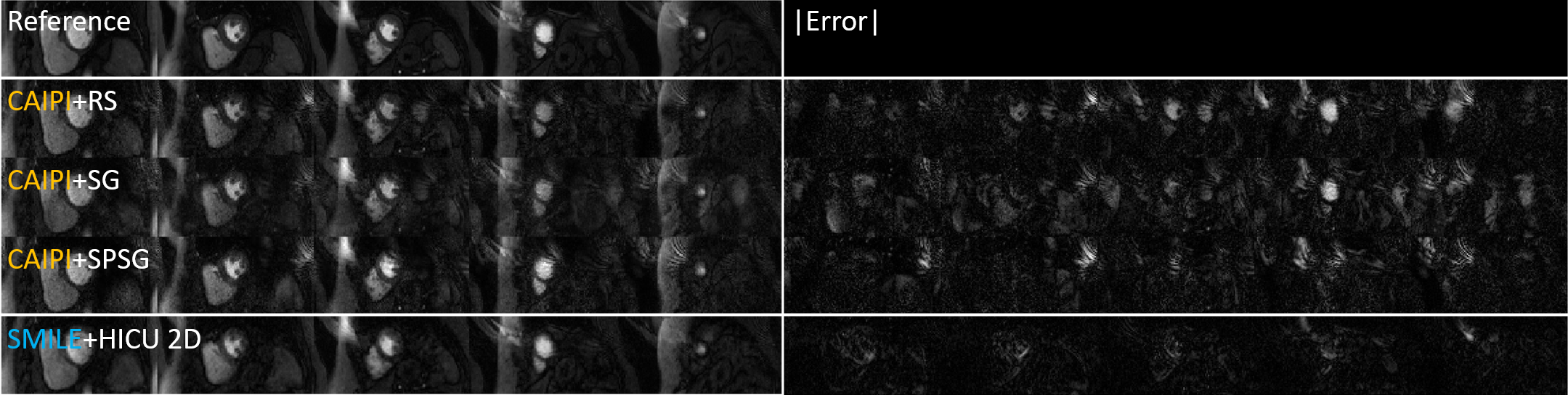}
    \caption{Representative retrospective result for MB factor = 5, $R$ = 5. The absolute error is scaled by 2, and windowed to match the reconstruction intensity range. }
    \label{fig:Retrospective MB=5}
\end{figure*}

\begin{figure*}
    \centering
    \includegraphics[width=\textwidth]{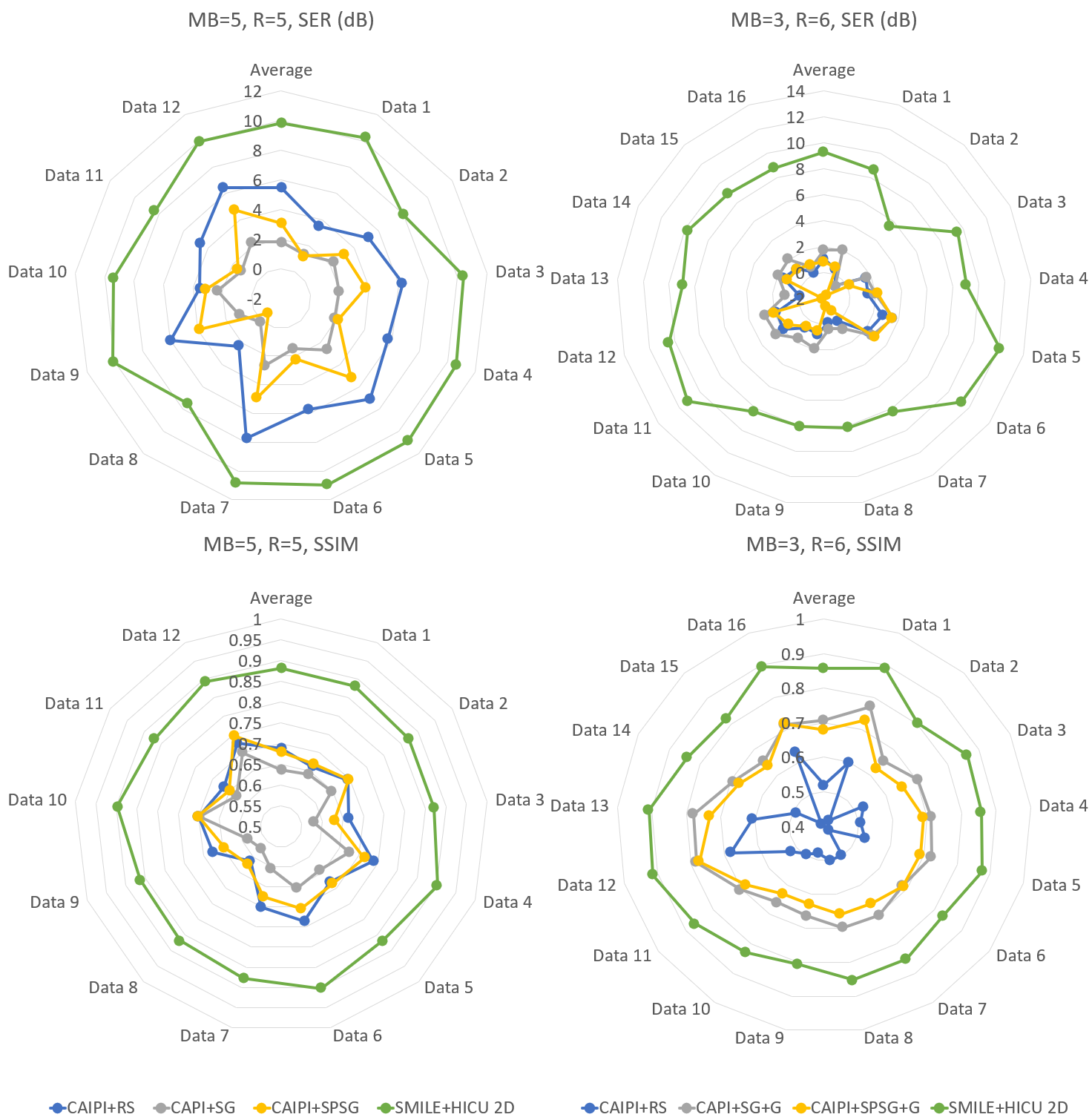}
    \caption{SER (Top), SSIM (Bottom) radar plots for MB = 5, $R$ = 5 (Left), MB = 3, $R$ = 6 (Right) experiments. Each colored dot represents one acquisition (CAIPI or SMILE) + reconstruction method combination for one dataset. The averaged performance is included in the plot as well. For both experiments SMILE + HICU significantly outperforms the CAIPI techniques.}
    \label{fig:SER SSIM}
\end{figure*}

\begin{figure*}
    \centering
    \includegraphics[width=\textwidth]{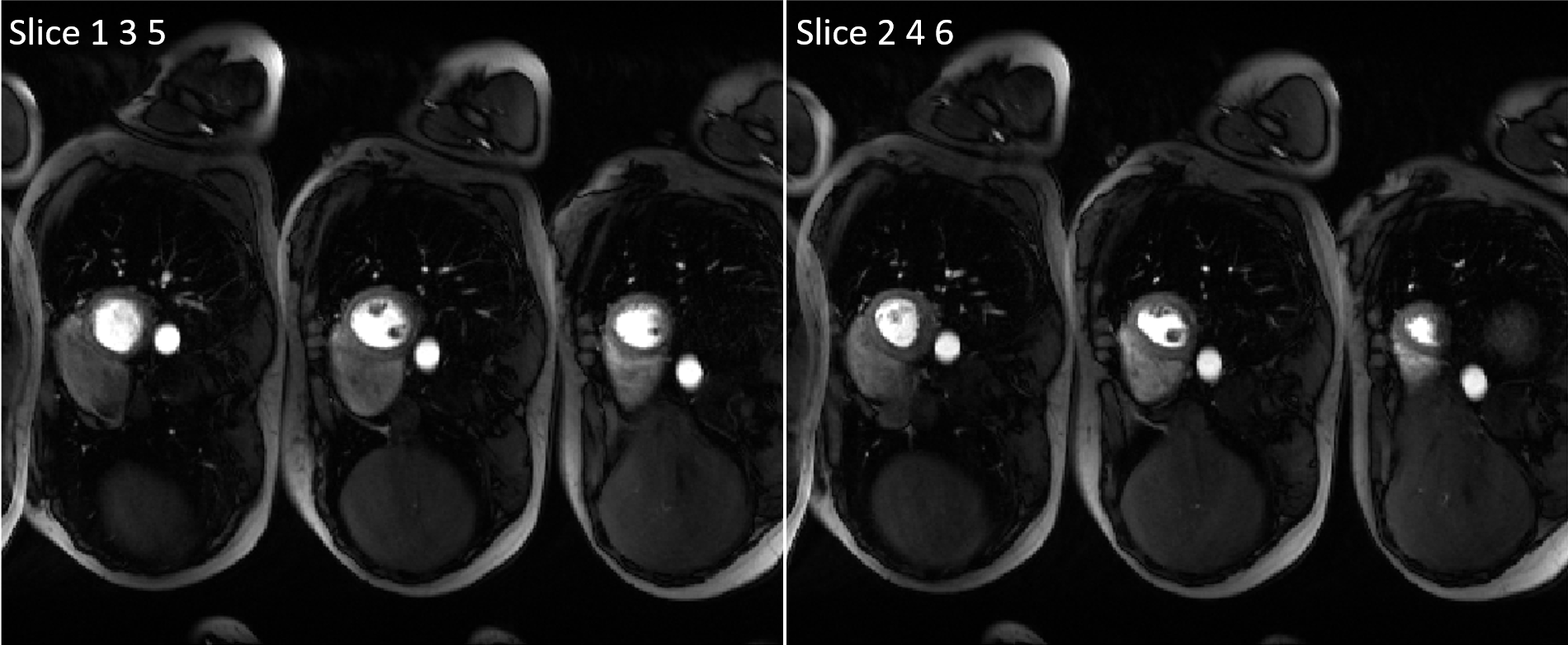}
    \caption{One representative prospective MB = 3, $R$ = 8 SMILE perfusion reconstructed by HICU + SWT. }
    \label{fig:Prospective MB=3}
\end{figure*}

\begin{figure*}
    \centering
    \includegraphics[width=\textwidth]{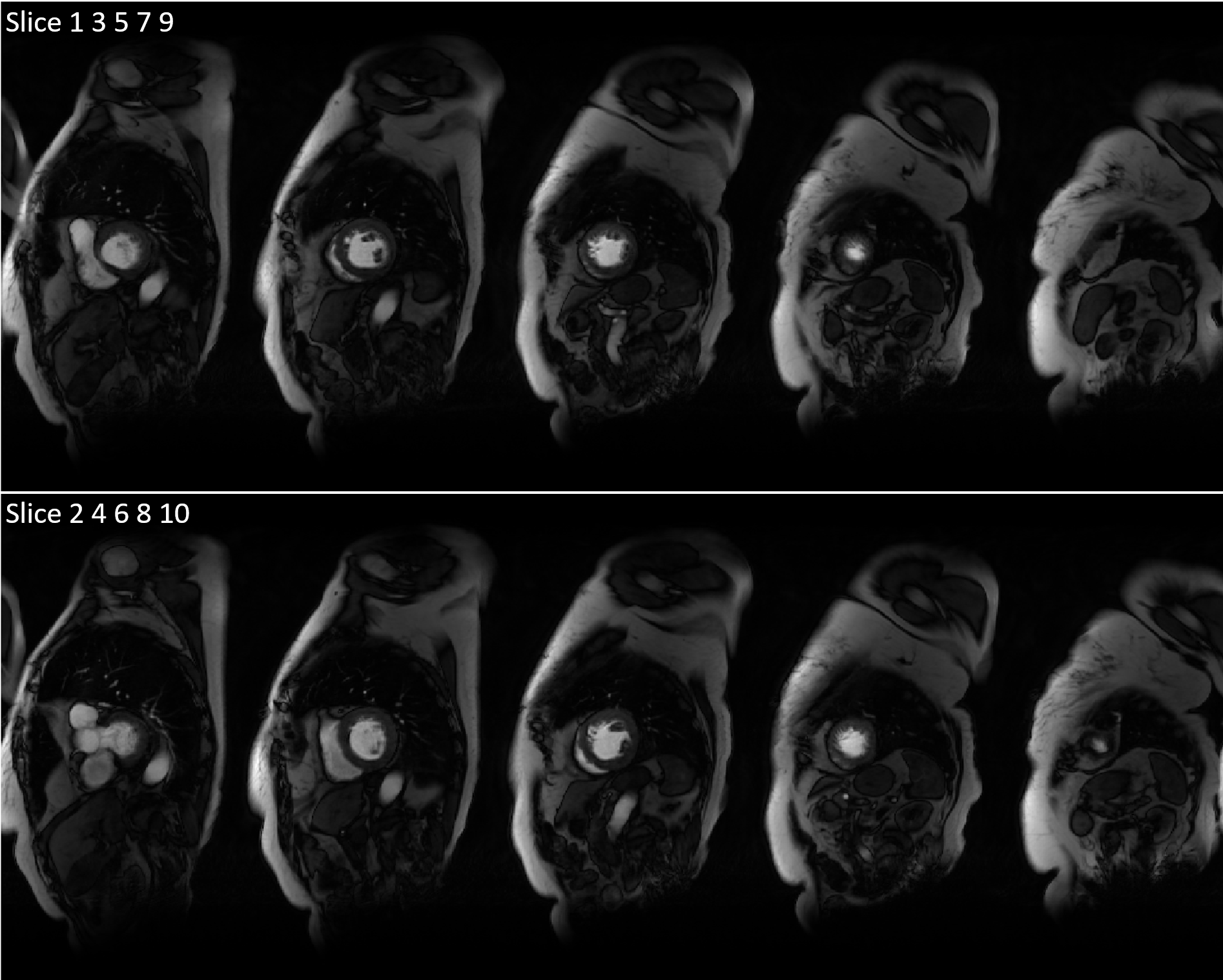}
    \caption{One representative prospective MB = 5, $R$ = 10 SMILE perfusion reconstructed by HICU + SWT.}
    \label{fig:Prospective MB=5}
\end{figure*}

\begin{figure*}
    \centering
    \includegraphics[width=\textwidth]{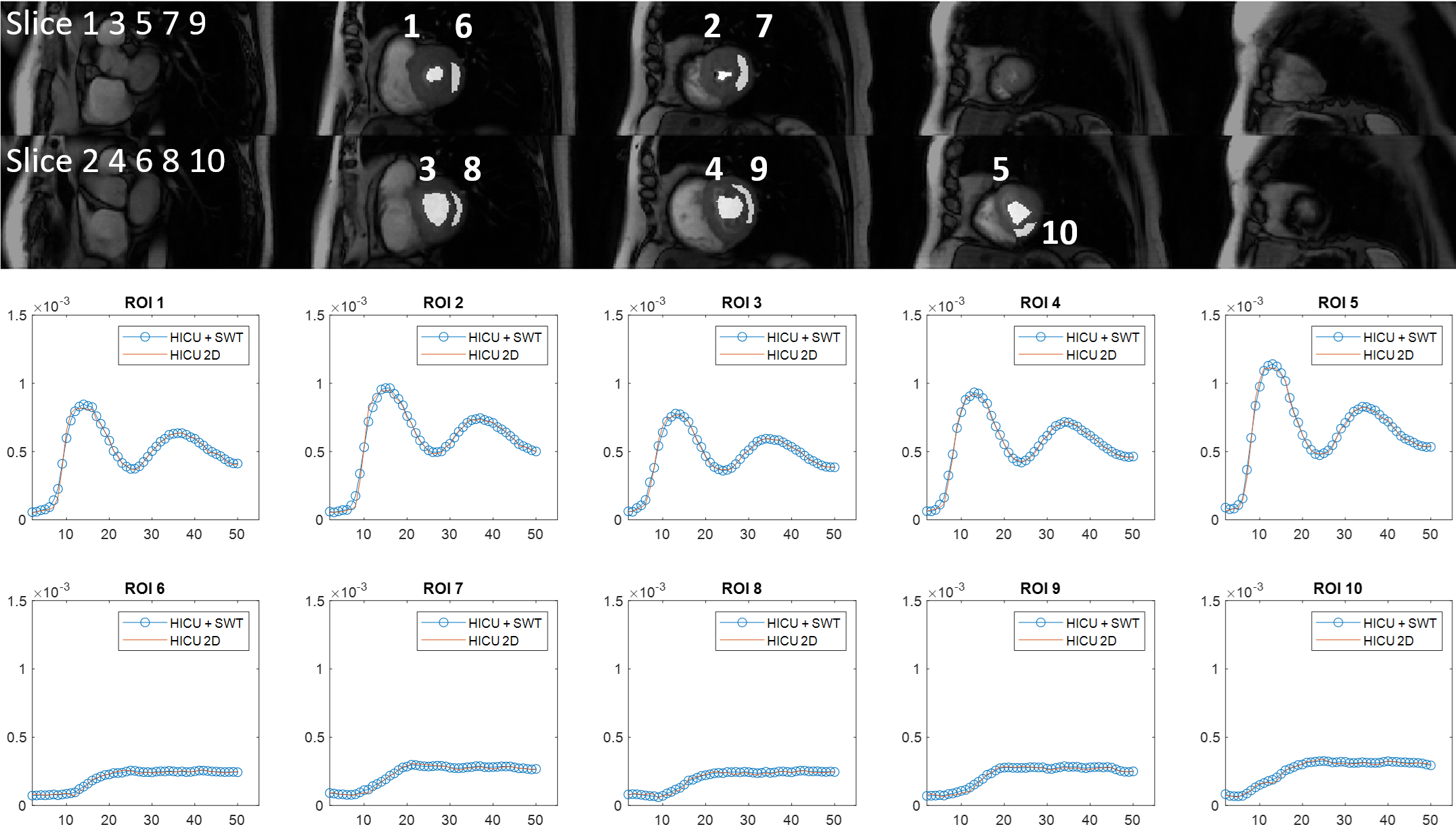}
    \caption{Representative temporal fidelity of ROIs for one prospective MB = 5, $R$ = 10 experiment. Top: ROIs are overlayed onto the reconstruction images. Bottom: HICU + SWT is compared against per frame HICU 2D reconstruction in terms of average signal intensity inside the ROIs. The proton density images acquired at the beginning of the acquisition are excluded in the curve plot here. We demonstrate that HICU + SWT, which includes temporal regularization, has good fidelity with HICU which does not have any temporal regularization.}
    \label{fig:Temporal fidelity}
\end{figure*}

\begin{figure*}
    \centering
    \includegraphics[width=\textwidth]{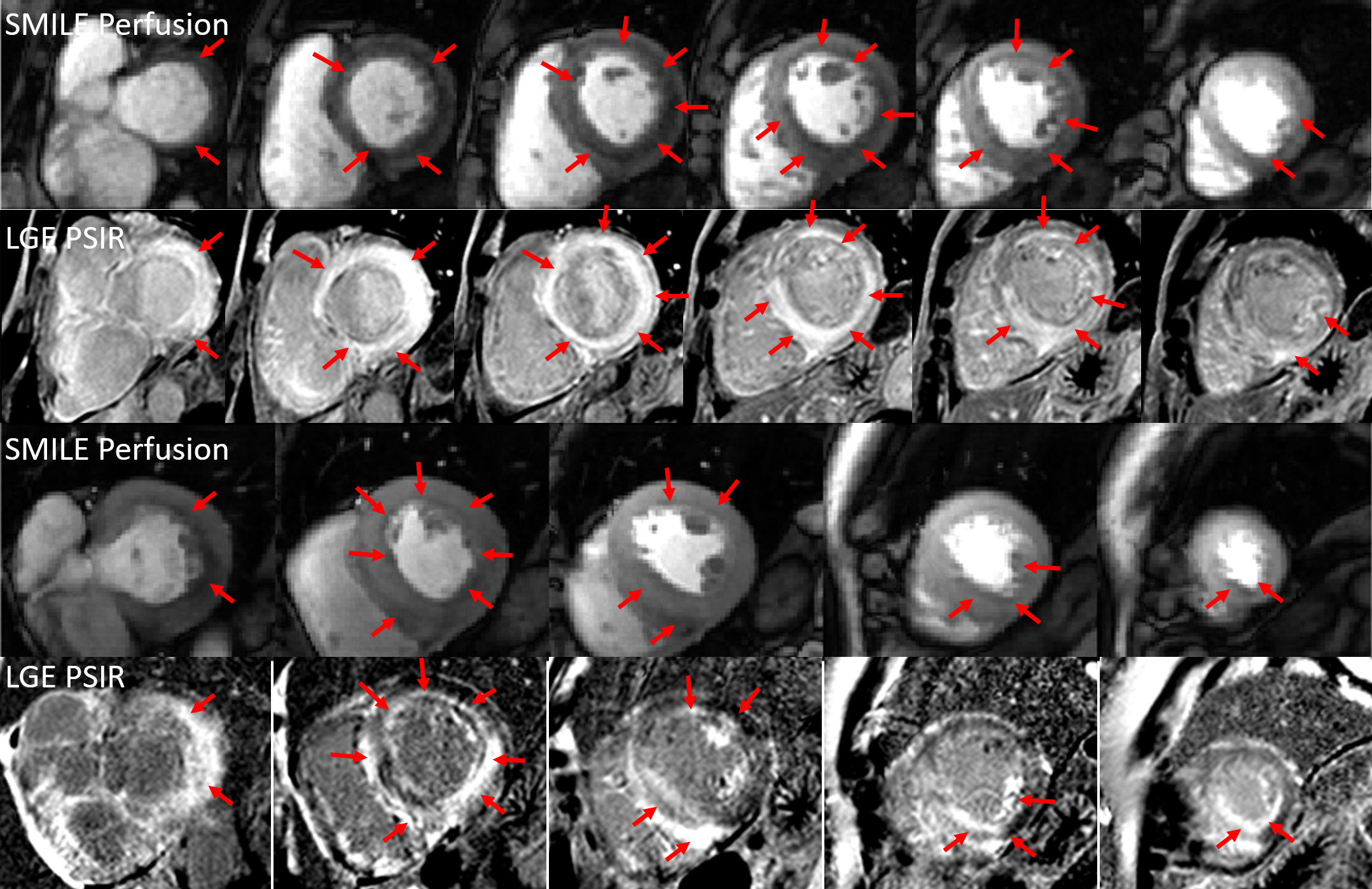}
    \caption{Two representative illustration of perfusion defect detected in rest MB = 6 and 5 within $5\times$ extended FOV, SMILE perfusion defect (Top) and scars validated in slices of single band LGE PSIR images (Bottom).}
    \label{fig:SMILE resting perfusion defect and LGE}
\end{figure*}

\section{Discussion}
From the retrospective experiments, improved image quality is observed in SMILE reconstruction both qualitatively and quantitatively.  The results demonstrate that the performance gain comes from both the acquisition strategy (SMILE vs CAIPI) as well as HICU as a higher performance reconstruction strategy. 

From the prospective experiment, high-quality perfusion images were obtained for MB of  3, with a 3$\times$ extended FOV and a net $R$ = 8, and for MB of 5, with a 5$\times$ extended FOV and net $R$ = 10.

Our proposed framework provides a comprehensive theoretical analysis of both acquisition and reconstruction. Compared to existing methods, such as CAIPI-based SMS perfusion \cite{nazir2018simultaneous} with MB = 2 and an in-plane acceleration rate of $R=5$, and \cite{yang2019whole} with MB up to 4 and an in-plane acceleration rate of $R=1.25$, as well as pseudo-random downsampled balanced SSFP SMS perfusion with extended FOV \cite{mcelroy2022simultaneous} at MB = 3 and net $R=5$, experimental results demonstrate that our proposed approach achieves high image quality with MB up to 5 and net $R$ values up to 10. The high MB enables all slices covering the left ventricle to be acquired during the same cardiac phase. This can have implications for quantitative perfusion imaging where there may be different biases in the perfusion values across the cardiac cycle.  Furthermore, if only one (5 slices) or two (10 slices) saturation blocks are required, each SMS block can be acquired over 250 ms, allowing for better sequence optimization and acquisition.  Of note, given the incoherent sampling adopted, the images appear free of dark-rim artifacts despite a longer temporal acquisition footprint.

The SMILE approach has additional advantages in that the extent of the extended FOV and the MB can be manipulated independently.  Thus, depending on extent of the region of interest and its location, within the FOV, a larger number of slices can be excited within a specific extended FOV.  For example, 6 slices could be simultaneously excited while only extending the FOV by a factor of 5, which is shown in first case of Fig.~\ref{fig:SMILE resting perfusion defect and LGE}. Additionally, the net $R$ can be determined by the desired total number of PE lines (e.g. the desired temporal footprint) independently of the size of the extended FOV or the MB factor.  Slices do not need to be uniformly shifted within the extended FOV but could exploit differences in the extent of the region of interest. As the standard CAIPI approach samples only a subset of possible locations, its sampling degrees of freedom is limited. In contrast, SMILE allows sampling across all PE lines in the extended FOV k-space, providing significantly greater flexibility. This enables more easily optimized sampling strategies compared to conventional SMS CAIPI.

The reconstruction methods specifically designed for CAIPI acquisition, SG and SPSG assume that each individual slice's k-space can be linearly predicted by the neighboring k-space of the superimposed k-space, where SPSG has an additional ``slice blocking'' component enforced in the prediction kernel extraction. RS reformulates the superimposed k-space into an undersampled k-space of extended FOV along the RO direction and then applied SPIRiT to reconstruct. One possible reason for the potential gain of RS (in MB = 5, $R$ = 5 of Fig.~\ref{fig:SER SSIM}) over SG or SPSG for the CAIPI acquisition is its adopted perspective of extended FOV in the RO dimension results in less bias in the kernel training than for the 3D perspective, which is intrinsically adopted in SG and SPSG. Another possible factor is that SPIRiT is better than GRAPPA in exploiting linear prediction from all sampled or non-sampled neighbor k-space points\cite{lustig2010spirit}.

The HICU reconstruction method adopted for SMILE is a fast general calibrationless method that utilize the k-space linear predication information among all dimensions and can be applied to both 2D, 2D + t, and other higher dimension cases. The main reason we adopt it as our main reconstruction choice for SMILE is its capability to reconstruct different imaging scenarios with mainly the kernel size and rank choice modification, competitive reconstruction quality, faster speed, and smaller memory footprint compared to other k-space low rank methods. Specifically, when it reduces to the 2D case, i.e., HICU 2D, its converged reconstruction quality is very similar to LORAKS \cite{haldar2013low} and SAKE \cite{shin2014calibrationless} but with one to two order of magnitude computation speedup (Can be seen in Table.~1 of reference \cite{zhao2021high}). Because they are essentially recovering the same convolutional operator given the same kernel size and rank, the speedup is coming from the leveraging Johnson-Lindenstrauss lemma to do the random projection of the null space, and center-out strategy to faster extract the null space information using a small portion of the center k-space.

We experimentally observed better reconstruction quality for SMILE acquisition over CAIPI acquisition which likely results from: 1. much larger sampling degree of freedom of SMILE and 2. intrinsically adopted 2D slice separation perspective in many CAIPI reconstructions leading to potential bias and ill-conditioning. The larger sampling degree of freedom is even more prominent when considering additional dimensions. e.g., in the Supporting Information, SWT can benefit SMILE + SENSE more than CAIPI + SENSE with an additional gain in SER for retrospective downsampled study ($p=0.003$) because SMILE's capability of allowing arbitrary PE sampling will further benefit incoherent sampling for the entire k-space and increase the performance gap.

In addition, as shown in the optimal tuned kernel size for RS, SMILE with HICU 2D and HICU, our previous analysis about the extended FOV effect on the k-space reconstruction kernel size is validated. The approximate coil image k-space support and the k-space reconstruction kernel size should increase proportionally to the FOV extended factor $n$. This might be the most counter-conventional practice that needs to be adopted for these two perspectives. We also noticed the kernel size for SG and SPSG for MB = 5 is the largest one which is limited by the calibration data size that makes the kernel weight deterministic. The larger kernel size is also consistent with the 3D perspective, where the inter-slice coil sensitivity map discontinuity makes the kernel size need to be larger to capture the coil sensitivity map information.

SMILE can enable slice-leakage-free perfusion imaging with full heart coverage. Cardiac perfusion MRI is an application, which will clearly benefit from this approach as the total amount of time data can be sampled is limited both by the RR interval duration and the transient first-pass passage of contrast. Other cardiac applications that would benefit are situations where for physiological reasons, there is an advantage of acquiring multiple slices simultaneously. There are situations where this may be advantageous for cine, LGE, or parametric mapping. 

The SMILE strategy is a general approach and can be used for other MRI applications. Fundamentally the strategy would also be compatible with Blipped-CAIPI\cite{setsompop2012blipped} and is not limited to RF-excitation modulation. There are likely more optimized sampling strategies and phase modulation strategies that could be employed using the SMILE approach.  This warrants further study.

Although spatial-temporal reconstruction improves SER and SSIM compared to 2D reconstruction (as shown in the Supporting Information), image quality can be compromised by temporal regularization, leading to blurring due to respiratory motion. This issue arises because clinical perfusion datasets were acquired during free breathing. To mitigate severe blurring, we avoid assuming strong spatial-temporal linear prediction and instead pair HICU 2D reconstruction with a light SWT denoiser. For prospectively acquired SMILE perfusion, data was acquired during breath-holding. For patients who can perform an adequate breath-hold maneuver during the acquisition, the spatial-temporal HICU + SWT can improve the reconstruction quality, but for other patients, it will slightly worsen the reconstruction. To ensure consistency across prospective experiments, we adopt spatial-temporal HICU + SWT reconstruction for all patients. The regularization introduced blurring is true for all spatial-temporal perfusion reconstructions, and not a specific limitation of SMILE.  

One limitation in our study is that our prospective study included only breathholding rest perfusion cases. We are currently exploring approaches to incorporating motion correction into the reconstruction as we and others have demonstrated previously \cite{yang2019whole, scannell2019robust, chen2024non}. Once motion correction is implemented, appropriate regularized spatial-temporal SMILE reconstruction should also have improved quality even during free breathing stress perfusion cases. Secondly, we only performed studies with rest perfusion, and thus cannot assess the ability of this technique to detect inducible perfusion abnormalities.  This will be the focus of our ongoing work on this technique. We demonstrate that SMILE can demonstrate resting perfusion abnormalities corresponding to LGE.  Furthermore, we demonstrate that the technique has high temporal fidelity for the myocardium and blood pool, which suggest that findings should be similar when applied during adenosine stress.

Another limitation of our study is that it was conducted at 3T and not at 1.5T. The lower SNR at 1.5T may affect performance \cite{oshinski2010cardiovascular}, particularly in high-acceleration settings, which is common for all MRI techniques when comparing 1.5T and 3T.  At lower field strength, the total acceleration factor may need to be decreased due to lower SNR. This effect could be mitigated for example by further exploring alternative acquisition strategies such as SSFP \cite{mcelroy2022simultaneous}. We plan to optimize the technique for both field strengths in future work.

\section{Conclusions}
In this work, we develop and implement the SMILE technique and apply it to cardiac perfusion to achieve high resolution and whole left ventricle coverage even at high multi-band factors. SMILE perfusion combined with existing parallel imaging and compressed sensing techniques, when appropriately adjusting ACS or k-space kernel sizes, achieves high image quality even at high total acceleration factors. It also demonstrates improved reconstruction quality without slice leakage, aligning with theoretical analysis.

\section{Data Availability Statement}
The code and example data are shared in \url{https://github.com/Zhao-Shen/SMILE}.


\newpage
\begin{thebibliography}{10}

\bibitem{canet1999magnetic}
Canet Emmanuelle~P, Janier Marc~F, Revel Didier. {Magnetic resonance perfusion imaging in ischemic heart disease}.  {\it J. Magn. Reson. Imaging. }1999;10(3):423--433.

\bibitem{otazo2010combination}
Otazo Ricardo, Kim Daniel, Axel Leon, Sodickson Daniel~K. Combination of compressed sensing and parallel imaging for highly accelerated first-pass cardiac perfusion MRI.  {\it Magn. Reson. Med.. }2010;64(3):767--776.

\bibitem{yang2019whole}
Yang Yang, Meyer Craig~H, Epstein Frederick~H, Kramer Christopher~M, Salerno Michael. Whole-heart spiral simultaneous multi-slice first-pass myocardial perfusion imaging.  {\it Magn. Reson. Med.. }2019;81(2):852--862.

\bibitem{kholmovski2007perfusion}
Kholmovski Eugene~G, DiBella Edward~VR. Perfusion MRI with radial acquisition for arterial input function assessment.  {\it Magn. Reson. Med.: An Official Journal of the International Society for Magn. Reson. Med.. }2007;57(5):821--827.

\bibitem{sharif2014towards}
Sharif Behzad, Dharmakumar Rohan, LaBounty Troy, et al. Towards elimination of the dark-rim artifact in first-pass myocardial perfusion MRI: removing Gibbs ringing effects using optimized radial imaging.  {\it Magn. Reson. Med.. }2014;72(1):124--136.

\bibitem{yang2016first}
Yang Yang, Kramer Christopher~M, Shaw Peter~W, Meyer Craig~H, Salerno Michael. First-pass myocardial perfusion imaging with whole-heart coverage using L1-SPIRiT accelerated variable density spiral trajectories.  {\it Magn. Reson. Med.. }2016;76(5):1375--1387.

\bibitem{shin2013three}
Shin Taehoon, Nayak Krishna~S, Santos Juan~M, Nishimura Dwight~G, Hu~Bob~S, McConnell Michael~V. Three-dimensional first-pass myocardial perfusion MRI using a stack-of-spirals acquisition.  {\it Magn. Reson. Med.. }2013;69(3):839--844.

\bibitem{wang2016radial}
Wang Haonan, Adluru Ganesh, Chen Liyong, Kholmovski Eugene~G, Bangerter Neal~K, DiBella Edward~VR. Radial simultaneous multi-slice CAIPI for ungated myocardial perfusion.  {\it Magn. Reson. Imaging. }2016;34(9):1329--1336.

\bibitem{mcelroy2022simultaneous}
McElroy Sarah, Ferrazzi Giulio, Nazir Muhummad~Sohaib, et al. Simultaneous multislice steady-state free precession myocardial perfusion with full left ventricular coverage and high resolution at 1.5 T.  {\it Magn. Reson. Med.. }2022;88(2):663--675.

\bibitem{sun2022slice}
Sun Changyu, Robinson Austin, Wang Yu, et al. A Slice-Low-Rank Plus Sparse (slice-L+S) Reconstruction Method for k-t Undersampled Multiband First-Pass Myocardial Perfusion MRI.  {\it Magn. Reson. Med.. }2022;88(3):1140--1155.

\bibitem{demirel2023signal}
Demirel Omer~Burak, Yaman Burhaneddin, Shenoy Chetan, Moeller Steen, Weing{\"a}rtner Sebastian, Ak{\c{c}}akaya Mehmet. Signal intensity informed multi-coil encoding operator for physics-guided deep learning reconstruction of highly accelerated myocardial perfusion CMR.  {\it Magn. Reson. Med.. }2023;89(1):308--321.

\bibitem{souza1988sima}
Souza SP, Szumowski J, Dumoulin CL, Plewes DP, Glover G. SIMA: simultaneous multislice acquisition of MR images by Hadamard-encoded excitation.  {\it J Comput Assist Tomogr. }1988;12(6):1026--1030.

\bibitem{glover1991phase}
Glover Gary~H. Phase-offset multiplanar ({POMP}) volume imaging: a new technique.  {\it J. Magn. Reson. Imaging. }1991;1(4):457--461.

\bibitem{larkman2001use}
Larkman David~J, Hajnal Joseph~V, Herlihy Amy~H, Coutts Glyn~A, Young Ian~R, Ehnholm G{\"o}sta. Use of multicoil arrays for separation of signal from multiple slices simultaneously excited.  {\it J. Magn. Reson. Imaging. }2001;13(2):313--317.

\bibitem{breuer2005controlled}
Breuer Felix~A, Blaimer Martin, Heidemann Robin~M, Mueller Matthias~F, Griswold Mark~A, Jakob Peter~M. Controlled aliasing in parallel imaging results in higher acceleration ({CAIPIRINHA}) for multi-slice imaging.  {\it Magn. Reson. Med.. }2005;53(3):684--691.

\bibitem{zahneisen2014three}
Zahneisen Benjamin, Poser Benedikt~A, Ernst Thomas, Stenger V~Andrew. Three-dimensional Fourier encoding of simultaneously excited slices: generalized acquisition and reconstruction framework.  {\it Magn. Reson. Med.. }2014;71(6):2071--2081.

\bibitem{zhu2016hybrid}
Zhu Kangrong, Dougherty Robert~F, Wu~Hua, et al. Hybrid-space SENSE reconstruction for simultaneous multi-slice MRI.  {\it IEEE Trans. Med. Imaging.. }2016;35(8):1824--1836.

\bibitem{blaimer2006accelerated}
Blaimer Martin, Breuer Felix~A, Seiberlich Nicole, et al. Accelerated volumetric MRI with a {SENSE/GRAPPA} combination.  {\it J. Magn. Reson. Imaging. }2006;24(2):444--450.

\bibitem{setsompop2012blipped}
Setsompop Kawin, Gagoski Borjan~A, Polimeni Jonathan~R, Witzel Thomas, Wedeen Van~J, Wald Lawrence~L. Blipped-controlled aliasing in parallel imaging for simultaneous multislice echo planar imaging with reduced g-factor penalty.  {\it Magn. Reson. Med.. }2012;67(5):1210--1224.

\bibitem{cauley2014interslice}
Cauley Stephen~F, Polimeni Jonathan~R, Bhat Himanshu, Wald Lawrence~L, Setsompop Kawin. Interslice leakage artifact reduction technique for simultaneous multislice acquisitions.  {\it Magn. Reson. Med.. }2014;72(1):93--102.

\bibitem{demirel2021improved}
Demirel Omer~Burak, Weing{\"a}rtner Sebastian, Moeller Steen, Ak{\c{c}}akaya Mehmet. Improved simultaneous multislice cardiac MRI using readout concatenated k-space SPIRiT (ROCK-SPIRiT).  {\it Magn. Reson. Med.. }2021;85(6):3036--3048.

\bibitem{zhao2024whole}
Zhao Shen, Wang Junyu, Salerno Michael. Whole Heart SMILE Perfusion with High Multiband Factor and High Acceleration Rate.  {\it J. Cardiovasc. Magn. Reson.. }2024;26.

\bibitem{feng2014golden}
Feng Li, Grimm Robert, Block Kai~Tobias, et al. {Golden-angle radial sparse parallel MRI: combination of compressed sensing, parallel imaging, and golden-angle radial sampling for fast and flexible dynamic volumetric MRI}.  {\it Magn. Reson. Med.. }2014;72(3):707--717.

\bibitem{shin2014calibrationless}
Shin Peter~J, Larson Peder~EZ, Ohliger Michael~A, et al. Calibrationless parallel imaging reconstruction based on structured low-rank matrix completion.  {\it Magn. Reson. Med.. }2014;72(4):959--970.

\bibitem{gungor2014subspace}
Gungor Derya~Gol. {\it Subspace techniques for parallel magnetic resonance imaging}.
\newblock The Ohio State University; 2014.

\bibitem{griswold2002generalized}
Griswold Mark~A, Jakob Peter~M, Heidemann Robin~M, et al. Generalized autocalibrating partially parallel acquisitions (GRAPPA).  {\it Magn. Reson. Med.. }2002;47(6):1202--1210.

\bibitem{lustig2010spirit}
Lustig Michael, Pauly John~M. SPIRiT: iterative self-consistent parallel imaging reconstruction from arbitrary k-space.  {\it Magn. Reson. Med.. }2010;64(2):457--471.

\bibitem{zhang2011parallel}
Zhang Jian, Liu Chunlei, Moseley Michael~E. Parallel reconstruction using null operations.  {\it Magn. Reson. Med.. }2011;66(5):1241--1253.

\bibitem{haldar2013low}
Haldar Justin~P. Low-rank modeling of local $ k $-space neighborhoods (LORAKS) for constrained MRI.  {\it IEEE transactions on medical imaging. }2013;33(3):668--681.

\bibitem{zhao2021high}
Zhao Shen, Potter Lee~C, Ahmad Rizwan. High-dimensional fast convolutional framework (HICU) for calibrationless MRI.  {\it Magn. Reson. Med.. }2021;86(3):1212--1225.

\bibitem{zhao2022high}
Zhao Shen, Ahmad Rizwan, Potter Lee. {\it High-dimensional fast convolutional framework ({HICU}) for calibrationless {MRI}. } US Patent App. 17/535,250; 2022.

\bibitem{uecker2014espirit}
Uecker Martin, Lai Peng, Murphy Mark~J, et al. ESPIRiT—an eigenvalue approach to autocalibrating parallel MRI: where SENSE meets GRAPPA.  {\it Magn. Reson. Med.. }2014;71(3):990--1001.

\bibitem{lobosadvanced}
Lobos Rodrigo~A, Kim Tae~Hyung, Setsompop Kawin, Haldar Justin~P. {\it {Advanced New Linear Predictive Reconstruction Methods for Simultaneous Multislice Imaging}. } 2020.

\bibitem{haldar2020linear}
Haldar Justin~P, Setsompop Kawin. Linear predictability in magnetic resonance imaging reconstruction: Leveraging shift-invariant Fourier structure for faster and better imaging.  {\it IEEE Signal Processing Magazine. }2020;37(1):69--82.

\bibitem{moeller2021diffusion}
Moeller Steen, Pisharady~Kumar Pramod, Andersson Jesper, et al. {Diffusion imaging in the post HCP era}.  {\it Journal of Magnetic Resonance Imaging. }2021;54(1):36--57.

\bibitem{chen2020ocmr}
Chen Chong, Liu Yingmin, Schniter Philip, et al. OCMR (v1. 0)--Open-Access Multi-Coil k-Space Dataset for Cardiovascular Magnetic Resonance Imaging.  {\it arXiv preprint arXiv:2008.03410. }2020;.

\bibitem{rich2020cartesian}
Rich Adam, Gregg Michael, Jin Ning, et al. CArtesian sampling with Variable density and Adjustable temporal resolution (CAVA).  {\it Magn. Reson. Med.. }2020;83(6):2015--2025.

\bibitem{robson2008comprehensive}
Robson Philip~M, Grant Aaron~K, Madhuranthakam Ananth~J, Lattanzi Riccardo, Sodickson Daniel~K, McKenzie Charles~A. {Comprehensive quantification of signal-to-noise ratio and g-factor for image-based and k-space-based parallel imaging reconstructions}.  {\it Magn. Reson. Med.. }2008;60(4):895--907.

\bibitem{nazir2018simultaneous}
Nazir Muhummad~Sohaib, Neji Radhouene, Speier Peter, et al. Simultaneous multi slice (SMS) balanced steady state free precession first-pass myocardial perfusion cardiovascular magnetic resonance with iterative reconstruction at 1.5 T.  {\it J. Cardiovasc. Magn. Reson.. }2018;20(1):84.

\bibitem{scannell2019robust}
Scannell Cian~M, Villa Adriana~DM, Lee Jack, Breeuwer Marcel, Chiribiri Amedeo. Robust non-rigid motion compensation of free-breathing myocardial perfusion MRI data.  {\it IEEE Trans. Med. Imaging.. }2019;38(8):1812--1820.

\bibitem{chen2024non}
Chen Quan, Wang Junyu, Wang Xitong, Zhao Shen, Liu Sizhuo, Salerno Michael. Non-rigid Motion Corrected Reconstruction with Nonlocally Low-rank Tensor Decomposition Based Regularization for High-resolution Cartesian First-pass Myocardial Perfusion Imaging at 3 Tesla.  {\it J. Cardiovasc. Magn. Reson.. }2024;26.

\bibitem{oshinski2010cardiovascular}
Oshinski John~N, Delfino Jana~G, Sharma Puneet, Gharib Ahmed~M, Pettigrew Roderic~I. Cardiovascular magnetic resonance at 3.0 T: current state of the art.  {\it J. Cardiovasc. Magn. Reson.. }2010;12(1):55.

\bibitem{jacob2020structured}
Jacob Mathews, Mani Merry~P, Ye~Jong~Chul. {Structured low-rank algorithms: Theory, magnetic resonance applications, and links to machine learning}.  {\it IEEE Signal Processing Magazine. }2020;37(1):54--68.

\bibitem{zhang2013coil}
Zhang Tao, Pauly John~M, Vasanawala Shreyas~S, Lustig Michael. {Coil compression for accelerated imaging with Cartesian sampling}.  {\it Magn. Reson. Med.. }2013;69(2):571--582.

\end{thebibliography}

\clearpage
\section{Supporting Information}
\subsection{SMILE Reconstruction Kernel Analysis}
Assuming the $c$-th continuous coil sensitivity map, $\bm{S}_c$, has a bandwidth size of $[B_x, B_y]$. As illustrated in Fig.~1 in the main document, the discrete k-space bandwidth of the coil sensitivity map for a single-slice FOV, with a sampling grid of $\Delta k_x, \Delta k_y$, amounts to $[C_x, C_y] = \big[\lfloor \tfrac{B_x}{\Delta k_x} \rfloor, \lfloor \tfrac{B_y}{\Delta k_y} \rfloor\big]$. For an $n \times$ extended FOV along PE, this value changes to $[C_x, D_y] = \big[\lfloor \tfrac{B_x}{\Delta k_x} \rfloor, \lfloor \tfrac{n B_y}{\Delta k_y} \rfloor\big]$.

Different sources \cite{haldar2020linear, jacob2020structured} such as smooth coil sensitivity maps, limited image support, transform domain sparsity, and exponential modeling give rise to the following relationship:
\begin{definition}[Approximate Linear Prediction]
    Each discrete k-space point $\rho[i]$ can be roughly expressed as a linear combination of appropriately chosen neighboring k-space points $\rho[i-j]$, combination weight $\omega[j]$ maintaining shift invariant. 
    \begin{equation}         
        \rho[i] = \sum_{j\in \Omega} \omega[j] \rho[i-j] + \epsilon ~\forall i, 
    \end{equation} 
    where $\Omega$ denotes indices of neighboring points, influencing the linear prediction kernel shape across dimensions like $k_x, k_y, z$ or $k_z$ and coil. $\epsilon$ implies approximation error due to additive noise and assumption deviations.
\end{definition}

Shift-invariant linear prediction equates to valid convolution across the dimensions involved, where convolution does not exceed boundaries. Linear prediction kernels can become annihilating kernels: the linear prediction kernel minus the extraction kernel forms an annihilating kernel \cite{haldar2020linear}, i.e.,
\begin{equation}    
    \sum_{j\in \Omega \cup \{0\}} \alpha[j] \rho[i-j] + \epsilon = 0  ~\forall i, \alpha[j] = 
    \begin{cases}
        \omega[j], &j \not = 0\\
        -1, &j = 0
    \end{cases}
\end{equation}   
Below is the sufficient condition for ideal linear prediction/annihilation of discrete k-space points solely based on a linearly independent perfect band-limited coil sensitivity map, where the linear independence refers to linear independent $\text{vec}(\bm{S}_c)$, achievable via coil compression \cite{zhang2013coil}. 

\begin{theorem}\label{thm1}
    Assume a rectangle of size $[C_x, C_y]$ bounds the discrete k-space support of $N_c$ linearly independent 2D coil sensitivity map $\mathcal{F}(\bm{S}_c)$. Then an annihilating kernel of size $[E_x, E_y, N_c]$ across $k_x, k_y$, coil exists given 
    \begin{equation}
        \label{eq: theorem 1}
        E_x E_y N_c > (C_x + E_x -1)(C_y + E_y -1).
    \end{equation}
\end{theorem}

\begin{proof}    
    \eqref{eq: theorem 1} ensures the count of unknowns (kernel elements) $>$ equation count (support size of convolution result), as attested in Appendix \cite{zhang2011parallel} and Section 3.2.1 \cite{gungor2014subspace}.
\end{proof}

If we relax $\bm{dom}(E_x), \bm{dom}(E_y)$ from $\mathbb{Z}_{>0}$ to $\mathbb{R}_{>0}$, by inequality of arithmetic and geometric means (AM-GM), and solving roots of a quadratic polynomial w.r.t. $\sqrt{E_xE_y}$, then the smallest kernel size for single-slice FOV is
\begin{equation}
    \label{eq: kernel size single slice FOV}
    \left[E_x^\star, E_y^\star\right] = \frac{1+\sqrt{N_c}}{N_c-1} \left[ (C_x-1), (C_y-1) \right].
\end{equation}
Likewise, for $n\times$ extended FOV in PE, the smallest kernel is
\begin{equation}
    \label{eq: kernel size extended FOV}
    \left[E_x^\star, E_y^\star\right] = \frac{1+\sqrt{N_c}}{N_c-1} \left[ (C_x-1), (D_y-1) \right].
\end{equation}
Comparing \eqref{eq: kernel size single slice FOV} and \eqref{eq: kernel size extended FOV}, the $E^\star_y$ increases by $\frac{D_y-1}{C_y-1}$. Although the result is for $E_x, E_y \in \mathbb{R}_{>0}$, the insight holds asymptotically for $E_x, E_y \in \mathbb{Z}_{>0}$ increases. For example, supposing $n \in \mathbb{Z}_{>0}$, $C_y = \lfloor \frac{D_y}{n} \rfloor >> 1$, if for single-slice FOV $[E_x^\star, E_y^\star] = [6,6]$, for $n\times$ FOV along PE, $[E_x^\star, E_y^\star] \approx [6,6n]$, which is rather atypical for 2D parallel imaging \cite{shin2014calibrationless, gungor2014subspace, griswold2002generalized, lustig2010spirit, zhang2011parallel,haldar2013low, zhao2021high, zhao2022high}. Using an insufficient kernel size can impact reconstruction performance and may have been a limitation to prior extended FOV techniques.

\subsection{2D+T CAIPI and SMILE reconstruction comparison}

\begin{figure*}
    \centering
    \includegraphics[width=\textwidth]{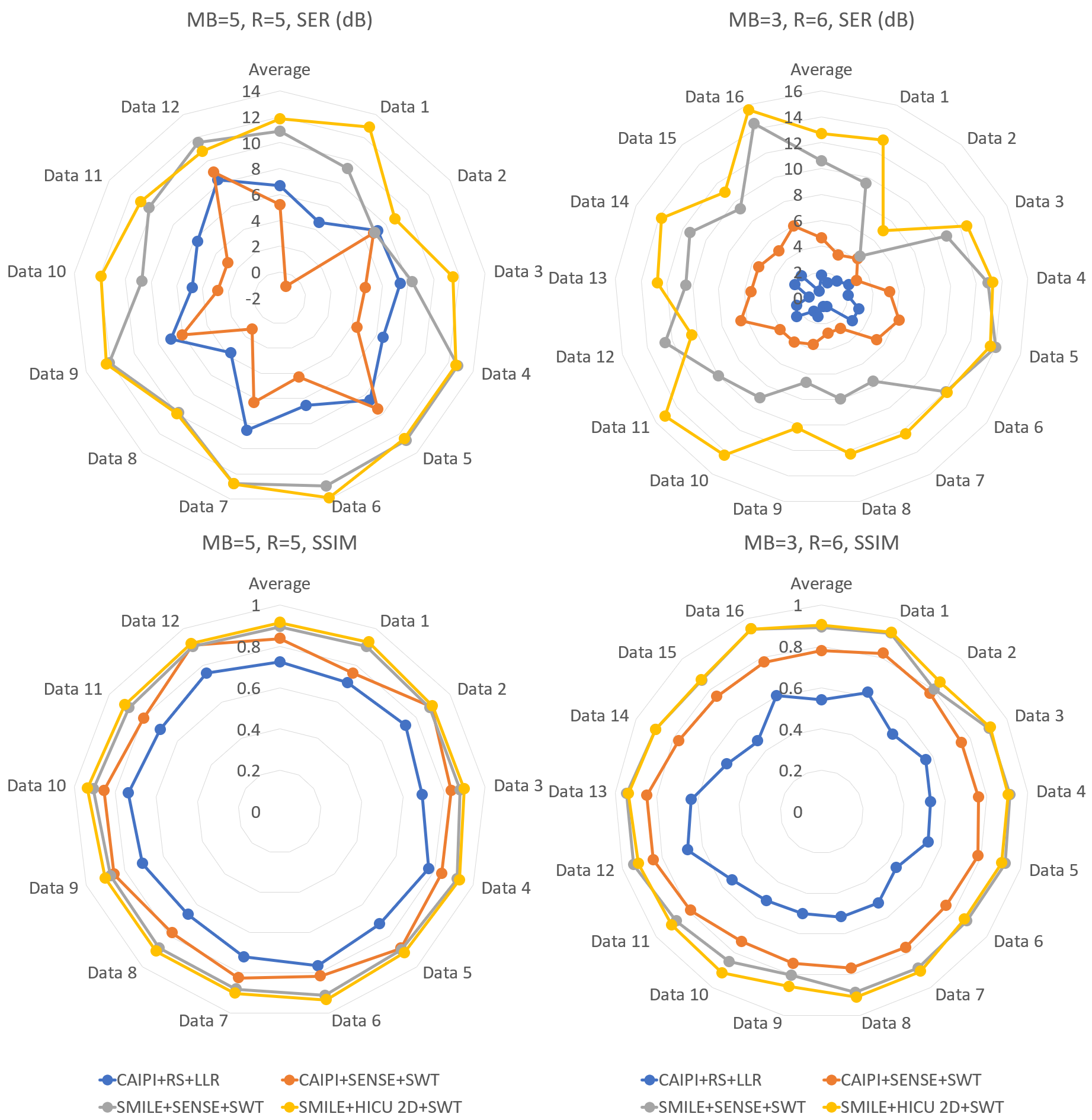}
    \caption{SER SSIM radar plots for spatial-temporal reconstruction for MB = 3, 5 experiments. In addition to the CAIPI/SMILE + 2D reconstruction technique, we also compare spatial-temporal reconstruction for CAIPI and SMILE, the regularization term includes locally low rank (LLR) \cite{demirel2021improved} and sparsity using stationary wavelet (SWT). If we choose and compare the combination with the best SER performance, for MB = 5, SMILE + HICU 2D + SWT is better than CAPI + RS + LLR by 5.2 dB, for MB = 3, SMILE + HICU 2D + SWT is better than CAIPI + SENSE + SWT by 8.0 dB, both show with statistically significant differences ($p<0.05$).}
    \label{fig:SER SSIM 2D + T}
\end{figure*}

\begin{figure*}    
    \centering
    \includegraphics[width=\textwidth]{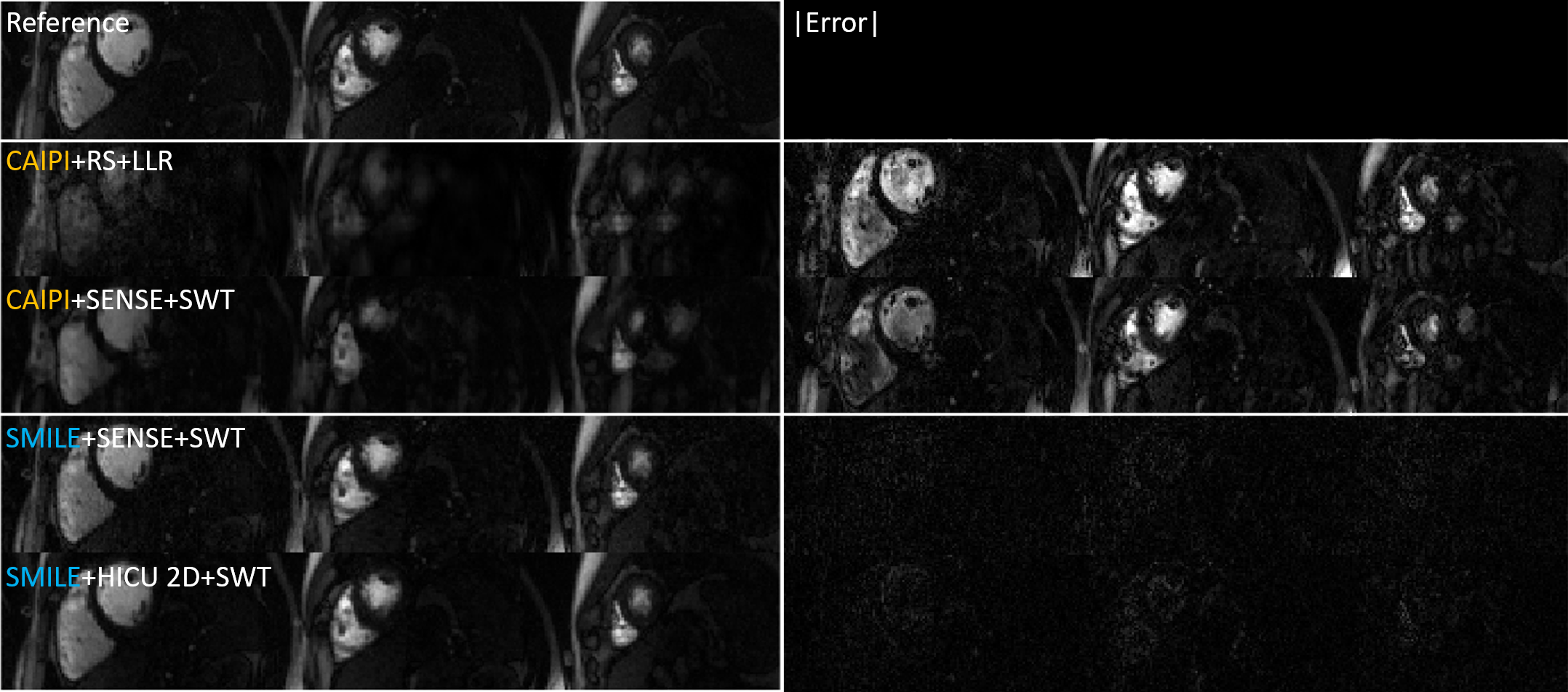}    
    \caption{Representative retrospective spatial-temporal reconstruction result for MB = 3, $R=6$ the absolute error is scaled by 2, and windowed to match the reconstruction intensity range. This corresponds to the same case shown in Fig.~4. }
    \label{fig:Retrospective MB=3 2D T}
\end{figure*}

\begin{figure*}
    \centering
    \includegraphics[width=\textwidth]{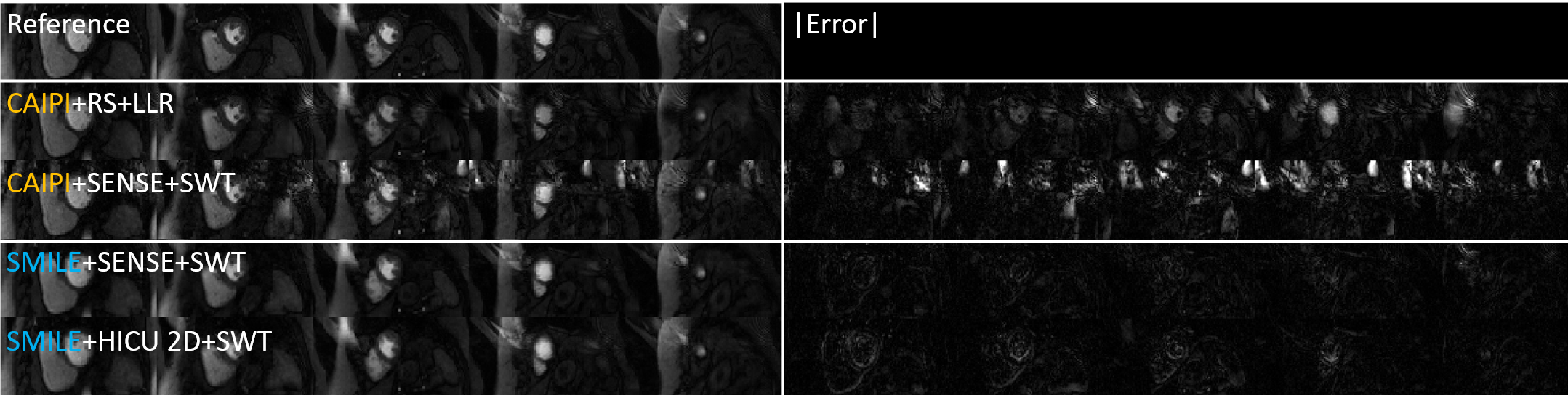}
    \caption{Representative retrospective spatial-temporal reconstruction result for MB = 5, $R=5$ the absolute error is scaled by 2, and windowed to match the reconstruction intensity range. This corresponds to the same case shown in Fig.~5.}
    \label{fig:Retrospective MB=5 2D T}
\end{figure*}
Sup.~Video.~V1. One representative prospective MB = 3, $R=8$, $3\times$ extended FOV, SMILE perfusion reconstructed by HICU + SWT. Left: slice 1 3 5, right: slice 2 4 6.

Sup.~Video.~V2. One representative prospective MB = 5, $R=10$, $5\times$ extended FOV, SMILE perfusion reconstructed by HICU + SWT. Up: slice 1 3 5 7 9, bottom: slice 2 4 6 8 10.

Sup.~Video.~V3. presents a representative retrospective result for MB = 3 and $R=6$. The absolute error is scaled by a factor of 2 and windowed to match the reconstruction intensity range. The video summarizes all acquisition + (2D or 2D + t) reconstruction combination for CAIPI and SMILE we have done. It corresponds to the same case shown in Fig.~4 and Fig.~S2.

Sup.~Video.~V4. presents a representative retrospective result for MB = 5 and $R=5$. The absolute error is scaled by a factor of 2 and windowed to match the reconstruction intensity range. The video summarizes all acquisition + (2D or 2D + t) reconstruction combination for CAIPI and SMILE we have done. It corresponds to the same case shown in Fig.~5 and Fig.~S3.

\end{document}